\documentclass[runningheads,a4paper]{llncs}

\usepackage{times}

\usepackage{amsmath}

\usepackage{amssymb}
\usepackage{array}
\usepackage{BeamerColor}
\usepackage{cite}
\usepackage{color}
\usepackage{empheq}
\usepackage{graphicx}
\usepackage{latexsym}
\usepackage{mathtools}
\usepackage{paralist}
\usepackage{relsize}
\usepackage{stmaryrd}
\usepackage{tabularx}
\usepackage{tikz}
\usetikzlibrary{arrows, automata, positioning, fit, patterns}
\usepackage[]{units}
\usepackage{url}
\usepackage{verbatim}
\usepackage{wasysym}
\allowdisplaybreaks

\newcommand{\lmerge}{\mathbin{
                  \setlength{\unitlength}{1ex}
                  \begin{picture}(1,1.75)
                  \put(0,0){\line(1,0){1}}
                  \put(0,0){\line(0,1){1.75}}
                  \put(0.45,0){\line(0,1){1.75}}
                  \end{picture}
                 }}
\newcommand{\cmerge}{~|~}
\newcommand{\hmerge}{\mathrel{|^{\negmedspace\scriptstyle /}}}
\newcommand{\textl}{~\!\lmerge}
\newcommand{\textc}{|}

\newcommand{\nil}{\mathbf{0}}
\newcommand{\Act}{\mathcal{A}}
\newcommand{\Acttau}{\mathcal{A}_\tau}
\newcommand{\Var}{\mathcal{V}}

\newcommand{\size}{\mathrm{size}}

\newcommand{\init}{\mathrm{init}}

\newcommand{\trans}[1][]{\xrightarrow{\, {#1} \, }}
\newcommand{\ntrans}[1][]{\mathrel{{\trans[#1]}\makebox[0em][r]{$\not$\hspace{2ex}}}{\!}}

\newcommand{\SOSrule}[2]{\frac{\displaystyle #1}{\displaystyle #2}}

\newcommand{\rel}{\,{\mathcal R}\,}

\newcommand{\Proc}{\mathbf{P}}

\newcommand{\FCCS}{\text{CCS}_{f}}

\newcommand{\E}{\mathcal{E}}

\newcommand{\tr}{\mathtt{T}}
\newcommand{\ctr}{\mathtt{CT}}
\newcommand{\fail}{\mathtt{F}}
\newcommand{\ready}{\mathtt{R}}
\newcommand{\ftr}{\mathtt{FT}}
\newcommand{\rtr}{\mathtt{RT}}
\newcommand{\pf}{\mathtt{PF}}

\newcommand{\s}{\mathtt{S}}
\newcommand{\cs}{\mathtt{CS}}

\newcommand{\rs}{\mathtt{RS}}
\newcommand{\B}{\mathtt{B}}
\newcommand{\wb}{\mathtt{WB}}
\newcommand{\rwb}{\mathtt{RWB}}

\usepackage[normalem]{ulem}

\let\doendexample\endexample
\renewcommand\endexample{~\hfill\scalebox{0.85}{$\LHD$}\doendexample} 

\newtheorem{asz}{Simplifying Assumption}
\newenvironment{assumption}{\begin{asz} \rm  }{\end{asz}}

\begin{document}

\title{In search of lost time: Axiomatising parallel composition in process algebras}

\author{
Luca Aceto\inst{1}\inst{2}
\and
Elli Anastasiadi\inst{2}
\and
Valentina Castiglioni\inst{2}
\and
Anna Ing{\'o}lfsd{\'o}ttir\inst{2}
\and
Bas Luttik\inst{3}
}

\institute{
Gran Sasso Science Institute, L'Aquila, Italy
\and
Reykjavik University, Reykjavik, Iceland 
\and
Eindhoven University of Technology, Eindhoven, The Netherlands
}

\authorrunning{L. Aceto et al.}

\maketitle

\begin{abstract}
This survey reviews some of the most recent achievements in the saga of the axiomatisation of parallel composition, along with some classic results. 
We focus on the recursion, relabelling and restriction free fragment of CCS and we discuss the solutions to three problems
that were open for many years. 
The first problem concerns the status of Bergstra and Klop’s auxiliary operators left merge and communication merge in the finite axiomatisation of parallel composition modulo bisimiliarity: 
We argue that, under some natural assumptions, the addition of a single auxiliary binary operator to CCS does not yield a finite axiomatisation of bisimilarity. 
Then we delineate the boundary between finite and non-finite axiomatisability of the congruences in van Glabbeek’s linear time-branching time spectrum over CCS. 
Finally, we present a novel result to the effect that rooted weak bisimilarity has no finite complete axiomatisation over CCS.
\end{abstract}


\section{Introduction}

\emph{Process algebras} \cite{BPS01,BBR10} are prototype specification languages for the description and analysis of concurrent and distributed systems, or simply \emph{processes}. 
These languages offer a variety of operators to specify composite processes from components one has already built. 
Notably, in order to model the concurrent interaction between processes, the majority of process algebras include some form of \emph{parallel composition} operator, also known as merge.

\emph{Equational axioms} can be used to express elegantly whether two syntactically different expressions in a language describe the same behaviour. 
Given the relevance of the parallel composition operator in concurrency theory, it is thus not surprising that the literature on process algebras offers a wealth of studies on the equational characterisation of this operator.

In this survey we review some classic results, and we briefly report on some of the most recent achievements on the axiomatisation of parallel composition.


\subsection{Motivation}

The expected semantic properties of processes can be defined, implicitly, by a set of \emph{equational axioms}, i.e., equalities between terms in the considered language.
Informally, if a term $t$ is proved equal to a term $u$ by means of the axioms, then we can say that $t$ and $u$ describe the same behaviour.

\emph{Behavioural equivalences} have been introduced as simple and elegant tools for an explicit comparison of the semantics of processes.
These are equivalence relations allowing one to establish whether two processes have the same \emph{observable behaviour}: different notions of observability correspond to different levels of abstraction from the information on process execution, which can either be considered irrelevant in an application context, or be unavailable to an external observer.
When an equivalence is compositional with respect to the operators in the considered language, it is called a \emph{congruence}.

Obtaining an axiomatisation of a behavioural congruence is then a classic, key problem in concurrency theory, as an equational axiomatisation characterises the semantics of a process algebra in a purely syntactic fashion.
Hence, this characterisation becomes independent of the details of the definition of the process semantics of interest.
At the same time, an axiomatisation also underlines the differences between the various semantics via a collection of revealing axioms, and, due to its syntactic nature, it can be applied in verification tools based on theorem proving of rewriting.

The studies on the axiomatisability of behavioural congruences find their roots in the attempt to provide an answer to the following two key questions:
\begin{itemize}
\item \emph{Is a collection of axioms \emph{complete}?}

This means that we are interested in verifying whether all the equations that hold modulo the chosen notion of behavioural congruence can be derived from the set of axioms using the rules of equational logic.

\item \emph{Does the process algebra modulo the chosen behavioural congruence afford a \emph{finite} equational axiomatisation?}

This means that we are interested in verifying whether there is a finite collection of axioms for the algebra that is sound and complete.
\end{itemize}

In this paper we focus on the parallel composition operator and present some of the most recent results on its equational characterisations modulo various behavioural congruences.
In order to put these new achievements into context, we present them alongside a brief recap of the saga of the axiomatisation of parallel composition, as without the work carried out in the last 40 years our new contributions would not have been possible.


\subsection{Standing on the shoulders of giants}

In the late 1970s, Milner developed the \emph{Calculus of Communicating Systems} (CCS) \cite{Mi80}, a formal language based on a message-passing paradigm and aimed at describing communicating processes from an operational point of view.
In detail, a \emph{labelled transition system} (LTS) \cite{Ke76} was used to equip language expressions with an \emph{operational semantics} \cite{Pl81} and was defined using a collection of syntax-driven rules.
The analysis of process behaviour was carried out via an observational \emph{bisimulation}-based theory \cite{Pa81} that defines when two states in an LTS describe the same behaviour.
In particular, CCS included a \emph{parallel composition operator} $\|$ to model the interactions among processes.
Such an operator, also known as \emph{merge} \cite{BK84b,BK85}, allows one both to \emph{interleave} the behaviours of its argument processes (modelling concurrent computations) and to enable some form of \emph{synchronisation} between them (modelling interactions).
Later on, in collaboration with Hennessy, Milner studied the \emph{equational  theory} of (recursion free) CCS and proposed a \emph{ground-complete axiomatisation} for it modulo bisimilarity \cite{HM80,HM85}. 
More precisely, Hennessy and Milner presented a set $\E$ of \emph{equational axioms} from which all equations over closed CCS terms (namely those with no occurrences of variables) that are \emph{valid modulo bisimilarity} can be derived using the rules of \emph{equational logic} \cite{T77}.
Notably, the set $\E$ included infinitely many axioms, which were instances of the \emph{expansion law} that was used to `simulate equationally' the operational semantics of the parallel composition operator.

The ground-completeness result by Hennessy and Milner started the quest for a finite axiomatisation of CCS's parallel composition operator modulo bisimilarity.

Bergstra and Klop showed in \cite{BK84b} that a finite ground-complete axiomatisation modulo bisimilarity can be obtained by enriching CCS with two auxiliary operators, namely the \emph{left merge} $\textl$ and the \emph{communication merge} $\textc$, expressing the one step in the asymmetric pure interleaving and the synchronous behaviour of $\|$, respectively.
Their result was then strengthened by Aceto et al.\ in \cite{AFIL09}, where it is proved that, over the fragment of CCS without recursion, restriction and relabelling, the auxiliary operators $\textl$ and $\textc$ allow for finitely axiomatising $\|$ modulo bisimilarity also when CCS terms with variables are considered.
Moreover, in \cite{AILT08} that result is extended to the fragment of CCS with relabelling and restriction, but without communication.
From those studies, we can infer that the left merge and communication merge operators are \emph{sufficient} to finitely axiomatise parallel composition modulo bisimilarity. 
But is the addition of auxiliary operators \emph{necessary} to obtain a finite equational axiomatisation, or can the use of the expansion law in the original axiomatisation of bisimilarity be replaced by a finite set of sound CCS equations?

To address that question, in \cite{Mo89,Mo90} Moller considered a minimal fragment of CCS, including only action prefixing, nondeterministic choice and interleaving, and proved that, even in the presence of a single action, bisimilarity does not afford a finite ground-complete axiomatisation over the closed terms in that language. This showed that auxiliary operators are indeed necessary to obtain a finite equational axiomatisation of bisimilarity.
Adapting Moller's proof technique, Aceto et al.\ proved, in \cite{AFIL05}, that if we replace $\textl$ and $\textc$ with the so called \emph{Hennessy's merge} $\hmerge$ \cite{He88}, which denotes an asymmetric interleaving with communication, then the collection of equations that hold modulo bisimilarity over the recursion, restriction and relabelling free fragment of CCS enriched with $\hmerge$ is not finitely based.
(\emph{Henceforth, whenever we refer to CCS, we mean the recursion, restriction and relabelling free fragment of CCS.})
A natural question that arises from those negative results is the following:
\begin{equation}
\tag{Q1}\label{eq:no_binary_intro}
\parbox{\dimexpr\linewidth-4em}{
\strut
\emph{Can we obtain a finite axiomatisation of parallel composition in bisimulation semantics by adding} 
only one binary operator
\emph{to the signature of CCS?}
\strut
}
\end{equation}
Recently, in \cite{ACFIL21}, we provided a \emph{partial} answer to that question: under three reasonable assumptions, a binary auxiliary operator alone does not allow us to obtain a finite, ground-complete axiomatisation modulo bisimilarity.
The results of \cite{ACFIL21} are discussed in Section~\ref{sec:no_binary}.

So far, we have considered equational characterisations of parallel composition modulo bisimilarity.
In \cite{vG90}, van Glabbeek presented the \emph{linear time-branching time spectrum}, i.e., a taxonomy of behavioural equivalences based on their distinguishing power.
He carried out his study in the setting of the process algebra BCCSP, which consists of the basic operators from CCS and CSP \cite{Ho85}, and he proposed \emph{ground-complete axiomatisations} for most of the congruences in the spectrum over this language.
Those axiomatisations are \emph{finite} if so is the set of actions.
For the ready simulation, ready trace and failure trace equivalences, the axiomatisation in \cite{vG90} made use of conditional equations;
Blom, Fokkink and Nain gave purely equational, finite axiomatisations in \cite{BFN03}.
Then, the works in \cite{AFvGI04}, on nested semantics, and in \cite{CFvG15}, on impossible futures semantics, completed the studies of the axiomatisability of behavioural congruences over BCCSP by providing \emph{negative} results: neither impossible futures nor any of the nested semantics have a finite, ground-complete axiomatisation over BCCSP. 
Since the algebra BCCSP does not include any operator for the parallel composition of processes, it is natural to ask:
\begin{equation}
\tag{Q2}\label{eq:spectrum_intro}
\parbox{\dimexpr\linewidth-4em}{
\strut
\emph{Can we obtain a finite axiomatisation of parallel composition, over CCS,} modulo congruences other \emph{than bisimilarity?}
\strut
}
\end{equation}
In Section~\ref{sec:spectrum} we report our results from \cite{ACILP20}, where we delineated the \emph{boundary} between finite and non-finite axiomatisability of the parallel composition operator modulo all the congruences in the spectrum.

As briefly outlined above, sometimes the information carried by the LTS can either be considered irrelevant or be unavailable to an external observer.
\emph{Weak behavioural semantics} have been introduced to study the effects of these unobservable (or \emph{silent}) actions on the observable behaviour of processes, each semantics considering a different level of abstraction.
A taxonomy of weak semantics is given in \cite{vG93}, and studies on the equational theories of various of these semantics have been carried out over the algebra BCCSP (see, among others, \cite{AdFGI14,CFvG08,dNH83,vGW96,HM85}).
A finite, ground-complete axiomatisation of parallel composition modulo \emph{rooted weak bisimilarity} (also known as \emph{observational congruence}) is provided by Bergstra and Klop in \cite{BK85} over the algebra ACP$\!\!\phantom{.}_\tau$ that, however, includes the auxiliary operators $\textl$ and $\textc$.
To the best of our knowledge, it is unknown whether the use of auxiliary operators is necessary to guarantee the finiteness of such an axiomatisation. 
Hence, the following question naturally arises: 
\begin{equation}
\tag{Q3}\label{eq:weak_intro}
\parbox{\dimexpr\linewidth-4em}{
\strut
\emph{Can we obtain a finite axiomatisation of parallel composition} modulo rooted weak bisimilarity \emph{over CCS}?
\strut
}
\end{equation}
In Section~\ref{sec:weak} we provide a negative answer to this question, by exploiting the results by Moller on strong semantics.
Actually, we obtain a stronger result: our proof technique applies to all congruences that coincide with strong bisimilarity on processes without silent moves, impose the root condition on initial silent moves, and satisfy a particular family of equations introduced by Moller in \cite{Mo90a}.
These include \emph{rooted branching bisimilarity}, \emph{rooted delay bisimilarity}, and \emph{rooted $\eta$-bisimilarity}.
We remark that this is an original contribution enriching our survey.


\subsection{Organisation of contents}
We start by reviewing some preliminary notions on process
algebras and equational logic in Section~\ref{sec:back}. 
We discuss some technical details in the proofs of classic results from in \cite{HM85,BK84b,Mo89} in Section~\ref{sec:classic}, since they are useful to appreciate our novel contributions. 
Then, our answers to questions (\ref{eq:no_binary_intro}), (\ref{eq:spectrum_intro}) and (\ref{eq:weak_intro}) are presented in Section~\ref{sec:no_binary}, Section~\ref{sec:spectrum}, and Section~\ref{sec:weak}, respectively.
We conclude our survey by discussing related and future work in Section~\ref{sec:conclusion}.


\section{Background}
\label{sec:back}

\subsection*{Labelled Transition Systems}
\label{Sect:lts+bis}

As semantic model we consider classic \emph{labelled transition systems} \cite{Ke76}.
We assume a set of action names $\Act$, and we let $\overline{\Act}$ denote the set of action co-names, i.e., $\overline{\Act}=\{\overline{a} \mid a \in \Act\}$.
As usual, we postulate that $\overline{\overline{a}}=a$ and $a \neq \overline{a}$ for all $a \in \Act$.
Then, we let $\Acttau = \Act \cup \overline{\Act} \cup \{\tau\}$.
Henceforth, we let $\mu,\nu,\dots$ range over actions in $\Acttau$, and $\alpha,\beta,\dots$ range over actions in $\Act \cup \overline{\Act}$.

\begin{definition}
[Labelled Transition System]
\label{Def:lts}
A {\sl labelled transition system} (LTS) is a triple $(\Proc,\Acttau,\trans[])$, where $\Proc$ is a set of \emph{processes} (or \emph{states}), $\Acttau$ is a set of {\sl actions}, and ${\trans[]} \subseteq \Proc \times \Acttau \times \Proc$ is a ({\sl labelled}) {\sl transition relation}. 
\end{definition}

As usual, we use $p \trans[\mu] p'$ in lieu of $(p,\mu,p') \in {\trans[]}$. 
For each $p \in \Proc$ and $\mu \in \Act$, we write $p \trans[\mu]$ if $p \trans[\mu] p'$ holds for some $p'$, and $p \ntrans[\mu]$ otherwise. 
The \emph{initials} of $p$ are the actions that label the outgoing transitions of $p$, that is, $\init(p) = \{\mu \mid p \trans[\mu] \}$. 
For a sequence of actions $\varphi = \mu_1 \cdots \mu_k$ ($k \geq 0$), and processes $p,p'$, we write $p \trans[\varphi] p'$ if and only if there exists a sequence of transitions $p = p_0 \trans[\mu_1] p_1 \trans[\mu_2] \cdots \trans[\mu_k] p_k = p'$. 
If $p \trans[\varphi] p'$ holds for some process $p'$, then $\varphi$ is a {\em trace} of $p$.


\subsection*{The Language CCS}
\label{sec:CCS}

The language we consider in this paper is the recursion, restriction and relabelling free fragment of Milner's CCS~\cite{Mi89}, given by the following grammar:
\[
t ::=\; \nil \;|\; 
x \;|\; 
\mu.t \;|\; 
t+t \;|\; 
t \mathbin{\|} t
\enspace ,
\]
where $x$ is a variable drawn from a countably infinite set $\Var$, and $\mu \in \Acttau$.
Following~\cite{Mi89}, the action symbol $\tau$ will result from the synchronised occurrence of two complementary actions, like $a$ and $\overline{a}$. 
The SOS rules \cite{Pl81} for the CCS operators given above are reported in Table~\ref{tab:sos_rules}.

We shall use the meta-variables $t,u,v,w$ to range over process terms.
The {\sl size} of a term $t$, denoted by $\size(t)$, is the number of operator symbols in $t$. 
A term is {\em closed} if it does not contain any variables.  
Closed terms, or {\sl processes}, will be denoted by $p,q,r$. 
Moreover, we omit trailing $\nil$'s from terms.

\begin{table}[t]
\begin{gather*}
\SOSrule{}{\mu.x \trans[\mu] x} 
\qquad
\SOSrule{x \trans[\mu] x'}{x + y \trans[\mu] x'} 
\qquad 
\SOSrule{y \trans[\mu] y'}{x + y \trans[\mu] y'} 
\\
\SOSrule{x \trans[\mu] x'}{x \mathbin{\|} y \trans[\mu] x' \mathbin{\|} y} 
\qquad
\SOSrule{y \trans[\mu] y'}{x \mathbin{\|} y \trans[\mu] x \mathbin{\|} y'} 
\qquad
\SOSrule{x \trans[\alpha] x' \quad y \trans[\overline{\alpha}] y'}{x \mathbin{\|} y \trans[\tau] x' \mathbin{\|} y'} 
\end{gather*}
\caption{\label{tab:sos_rules} The SOS rules for CCS operators ($\mu \in \Acttau$, $\alpha \in \Act \cup \overline{\Act}$).} 
\end{table}

A {\sl (closed) substitution} is a mapping from process variables to (closed) CCS terms.
Substitutions are extended from variables to terms, transitions, and rules in the usual way.
Note that $\sigma(t)$ is closed, if so is $\sigma$. 
We let $\sigma[x\mapsto p]$ denote the substitution that maps the variable $x$ into process $p$ and behaves like $\sigma$ on all other variables.

The inference rules in Table~\ref{tab:sos_rules} allow us to derive valid transitions between closed terms.
The operational semantics for our language (and for all its extensions we shall introduce in the remainder of the paper) is then modelled by the LTS whose processes are the closed terms, and whose labelled transitions are those that are provable from the SOS rules.

In the remainder of this paper, we exploit the associativity and commutativity of $+$ modulo bisimilarity and we consider process terms modulo them, namely we do not distinguish $t+u$ and $u+t$, nor $(t+u)+v$ and $t+(u+v)$.  
In what follows, the symbol $=$ will denote equality modulo the above identifications. 
We use a {\em summation} $\sum_{i\in\{1,\ldots,k\}}t_i$ to denote the term $t= t_1+\cdots+t_k$, where the empty sum represents $\nil$.
We can also assume that the terms $t_i$, for $i \in \{1,\dots,k\}$, do not have $+$ as head operator, and refer to them as the \emph{summands} of $t$.

Henceforth, for each action $\mu$ and $m \ge 0$, we let $\mu^0$ denote $\nil$ and $\mu^{m+1}$ denote $\mu(\mu^m)$.  


\subsection*{Behavioural equivalences}

\emph{Behavioural equivalences} have been introduced as a tool to establish whether the behaviours of two processes are \emph{indistinguishable for their observers}.
Roughly, they allow us to check whether the \emph{observable} semantics of two processes is \emph{the same}.
In the literature we can find several notions of behavioural equivalence based on the observations that an external observer can make on the process.
In his seminal article \cite{vG90}, van Glabbeek gave a taxonomy of the behavioural equivalences discussed in the literature on concurrency theory, which is now called the \emph{linear time-branching time spectrum} (see Figure~\ref{fig:spectrum}).
 
One of the main concerns in the development of a meta-theory of process languages is to guarantee their \emph{compositionality}, i.e., that the \emph{replacement} of a component of a system with a $\sim$-equivalent one, for a chosen behavioural equivalence $\sim$, does not affect the behaviour of that system.
In algebraic terms, this is known as the \emph{congruence property} of $\sim$ with respect to all language operators, and consists in verifying whether, for all $n$-ary operators $f$
\[
t_i \sim t_i', \;\forall\, i = 1,\dots,n \Longrightarrow f(t_1,\dots,t_n) \sim f(t_1',\dots,t_n').
\]
Since all the operators considered in this paper are defined by inference rules in the de Simone format \cite{dS85}, by \cite[Theorem 4]{vG93b} we have that all the equivalences in the spectrum in Figure~\ref{fig:spectrum} are congruences with respect to them.

In particular, we shall consider the states in a LTS modulo the equivalence relation known as bisimilarity~\cite{Mi89,Pa81}.

\begin{definition}
[Bisimilarity]
\label{Def:bisimulation}
Let $(\Proc,\Act,\trans[])$ be a LTS. 
{\sl Bisimilarity}, denoted by $\sim_\B$, is the largest binary symmetric relation over $\Proc$ such that whenever $p \sim_\B q$ and $p \trans[\mu] p'$, then there is a transition $q \trans[\mu] q'$ with $p' \sim_\B q'$.
If $p \sim_\B q$, then we say that $p$ and $q$ are {\sl bisimilar}. 
\end{definition}


\subsection*{Equational Logic}
\label{Sect:logic}

An \emph{axiom system} $\E$ is a collection of (\emph{process}) \emph{equations} $t \approx u$ over CCS.
An equation $t \approx u$ is \emph{derivable} from an axiom system $\E$, notation $\E \vdash t \approx u$, if there is an \emph{equational proof} for it from $\E$, namely if $t \approx u$ can be inferred from the axioms in $\E$ using the \emph{rules} of \emph{equational logic}.
The rules over CCS are reported in Table~\ref{tab:equational_logic}. 
\begin{table}[t]
\begin{gather*}
\scalebox{0.9}{($e_1$)}\; t \approx t 
\qquad
\scalebox{0.9}{($e_2$)}\; \frac{t \approx u}{u \approx t} 
\qquad
\scalebox{0.9}{($e_3$)}\; \frac{{t \approx u ~~ u \approx v}}{{t \approx v}} 
\qquad
\scalebox{0.9}{($e_4$)}\; \frac{{t \approx u}}{{\sigma(t) \approx \sigma(u)}} 
\\[.2cm]
\scalebox{0.9}{($e_5$)}\; \frac{t \approx u}{\mu. t \approx \mu. u}
\qquad
\scalebox{0.9}{($e_6$)}\; \frac{t \approx  u~~ t' \approx u'}{t+t' \approx u+u'}
\qquad
\scalebox{0.9}{($e_7$)}\; \frac{t \approx  u~~ t' \approx u'}{t\mathbin{\|} t' \approx u\mathbin{\|} u'}
\enspace .
\end{gather*}
\caption{\label{tab:equational_logic} The rules of equational logic} 
\end{table}
Without loss of generality one may assume that substitutions happen first in equational proofs, i.e., that rule ($e_4$) may only be used when $(t \approx u) \in \E$.  
In this case $\sigma(t) \approx \sigma(u)$ is called a {\em substitution instance} of an axiom in $\E$.

We are interested in equations that are valid modulo some congruence relation $\sim$ over closed terms.
The equation $t \approx u$ is said to be \emph{sound} modulo $\sim$ if $\sigma(t) \sim \sigma(u)$ for all closed substitutions $\sigma$.
For simplicity, if $t \approx u$ is sound, then we write $t \sim u$.
An axiom system is \emph{sound} modulo $\sim$ if, and only if, all of its equations are sound modulo $\sim$. 
Conversely, we say that $\E$ is \emph{complete} modulo $\sim$ if $t \sim u$ implies $\E \vdash t \approx u$ for all terms $t,u$.
If we restrict ourselves to consider only equations in which there are no occurrences of variables, then $\E$ is said to be \emph{ground-complete} modulo $\sim$.
We say that $\sim$ has a \emph{finite}, (ground) complete, axiomatisation, if there is a \emph{finite} axiom system $\E$ that is sound and (ground) complete for $\sim$.


\section{Classic results}
\label{sec:classic}

In this section we give a bird's eye view on the classic results on the axiomatisability of the parallel composition operator that we surveyed in the Introduction, in order to underline some technical features that will be useful in the remainder of the paper. 

\begin{table}[t]
\centering
\begin{tabular}{ll}
\, \scalebox{0.95}{(A0)} \; $x + \nil \approx x$ & 
\qquad \scalebox{0.95}{(P0)} \; $x \mathbin{\|} \nil \approx x$ \\[.1cm]
\, \scalebox{0.95}{(A1)} \; $x + y \approx y + x$ & 
\qquad \scalebox{0.95}{(P1)} \; $x \mathbin{\|} y \approx y \mathbin{\|} x$ \\[.1cm]
\, \scalebox{0.95}{(A2)} \; $(x + y) + z \approx x + (y+z)$ \\[.1cm]
\, \scalebox{0.95}{(A3)} \; $x + x \approx x$ \\[.2cm]
\end{tabular}
\caption{Basic axioms for CCS.
We define $\E_0 = \{\mathrm{A0,A1,A2,A3}\}$ and $\E_1 = \E_0 \cup \{\mathrm{P0,P1}\}$.}
\label{tab:basic_axioms}
\end{table}

The first study on the equational characterisation of parllel composition was carried out by Hennessy and Milner, in their seminal paper \cite{HM80} (preliminary version of \cite{HM85}).
There, they provided a ground-complete axiomatisation of CCS modulo bisimilarity.
This axiomatisation consisted of the axiom system $\E_0 = \{\mathrm{A0,A1,A2,A3}\}$ given in Table~\ref{tab:basic_axioms}, which is a ground-complete axiomatisation of BCCSP modulo bisimilarity (a proof can be found, e.g., in \cite{vG90}), enriched with the axiom schema EL in Table~\ref{tab:exp_law}, known as \emph{the expansion law}.

\begin{table*}[t]
\[
\begin{array}{ll}
\, \scalebox{0.95}{(EL)} & \displaystyle \sum_{i \in I} \mu_ix_i \mathbin{\|} \sum_{j \in J} \nu_j y_j \approx \sum_{i \in I} \mu_i (x_i \mathbin{\|} \sum_{j \in J} \nu_j y_j) + 
\sum_{j \in J} \nu_j (\sum_{i \in I} \mu_i x_i \mathbin{\|} y_j) + \sum_{i \in I, j \in J \atop \mu_i = \overline{\nu_j}} \tau(x_i \mathbin{\|} y_j) \\[.2cm]
\, \scalebox{0.95}{(EL1)} & \mu x \mathbin{\|} \nu y \approx \mu(x \mathbin{\|} \nu y) + \nu(\mu x \mathbin{\|} y)
\phantom{\;\, + \tau(x \mathbin{\|} y)}
\qquad \text{ if } \mu \neq \overline{\nu}, 
\text{ or } \mu = \tau, 
\text{ or } \nu = \tau \\[.1cm]
\, \scalebox{0.95}{(EL1$\tau$)} & \mu x \mathbin{\|} \nu y \approx \mu(x \mathbin{\|} \nu y) + \nu(\mu x \mathbin{\|} y) + \tau(x \mathbin{\|} y)
\qquad \text{ if } \mu = \overline{\nu} \\[.2cm]
\, \scalebox{0.95}{(EL2)} & \displaystyle \sum_{i \in I} \mu_ix_i \mathbin{\|} \sum_{j \in J} \nu_j y_j \approx \sum_{i \in I} \mu_i (x_i \mathbin{\|} \sum_{j \in J} \nu_j y_j) + 
\sum_{j \in J} \nu_j (\sum_{i \in I} \mu_i x_i \mathbin{\|} y_j) +
\sum_{i \in I, j \in J \atop \mu_i = \overline{\nu_j}} \tau(x_i \mathbin{\|} y_j) \\
& \text{ with } \mu_i \neq \mu_k \text{ if } i \neq k \text{ and } \nu_j \neq \nu_h \text{ if } j \neq h, \forall\, i,k \in I, \forall\, j,h \in J \\[.2cm]
\end{array}
\]
\caption{\label{tab:exp_law} The different instantiations of the expansion law.}
\end{table*}

The expansion law was used to deal with parallel composition: it states that whenever the initial behaviour of the two parallel components is known, then the initial behaviour of their composition can be described explicitly by the term on the right-hand side of equation EL.
Informally, as parallel composition does not distribute over choice in either of its arguments, modulo bisimilarity, the only way to describe equationally the initial behaviour of a term of the form $p \mathbin{\|} q$  is first to express $p$ as $\sum_{i \in I} p_i$ and $q$ as $\sum_{j \in J} q_j$. 
One can then apply the expansion law, from left to right, to eliminate all occurrences of parallel composition from CCS processes, reducing them to BCCSP processes, and then use the ground-completeness of $\E_0$ over BCCSP to conclude that $\E_0 \cup \{EL\}$ is ground-complete over CCS modulo bisimilarity.

\begin{theorem}
[Hennessy and Milner \cite{HM80,HM85}]
The axiom system $\E_0 \cup \{EL\}$, including all possible instances of $EL$, is a ground-complete axiomatisation of CCS modulo bisimilarity.
\end{theorem}

However, the axiomatisation proposed by Hennessy and Milner was \emph{not finite}.
In fact, the axiom schema EL generates \emph{infinitely many axioms}, even if the set of actions over which CCS terms are built is finite.

\begin{table}[t]
\begin{gather*}
\SOSrule{x \trans[\mu] x'}{x \lmerge y \trans[\mu] x' \mathbin{\|} y} 
\qquad
\SOSrule{x \trans[\alpha] x' \quad y \trans[\overline{\alpha}] y'}{x \cmerge y \trans[\tau] x' \mathbin{\|} y'} 
\end{gather*}
\caption{\label{tab:merge_rules} The SOS rules for $\textl$ and $\textc$ ($\mu \in \Acttau, \alpha \in \Act \cup \overline{\Act}$).} 
\end{table}

The first \emph{finite} ground-complete axiomatisation of parallel composition was proposed by Bergstra and Klop in \cite{BK84b}.
It was obtained by introducing two auxiliary operators to the syntax of CCS, the \emph{left merge} ($\textl$) and the \emph{communication merge} ($\cmerge$) operators, whose operational semantics is given in Table~\ref{tab:merge_rules}.
The left merge operator expresses the one step asymmetric pure interleaving behaviour of $\mathbin{\|}$, while the communication merge operator captures the synchronous behaviour of $\mathbin{\|}$.
Hence, we can use these two auxiliary operators to express equationally the behaviour of parallel composition as follows:
\begin{equation}
x \mathbin{\|} y \approx x \lmerge y + y \lmerge x + x \cmerge y
\tag{PE}\label{eq:parallel_expansion} 
\end{equation}
Intuitively, since $\textl$ distributes over choice in its first argument and $\textc$ does it in both arguments, thanks to axiom~\ref{eq:parallel_expansion} we can reduce a process of the form $\sum_{i \in I} p_i \mathbin{\|} \sum_{j \in J} q_j$ to a process of the form $\sum_{i \in I, j \in J} (p_i \mathbin{\|} q_j)$ and we can then proceed to eliminate $\mathbin{\|}$ from these simpler summands.
Hence, we no longer need to use EL to expand the initial behaviour of a parallel composition of sums, and we can thus obtain:

\begin{theorem}
[Bergstra and Klop, \cite{BK84b}]
There exists a \emph{finite}, ground-complete axiomatisation of CCS enriched with $\textl$ and $\textc$, modulo bisimilarity.
\end{theorem}

Later on, Moller proved that the use of auxiliary operators is not only sufficient to obtain a finite equational charasterisation of $\mathbin{\|}$, but it is \emph{necessary} indeed.

\begin{theorem}
[Moller \cite{Mo89,Mo90,Mo90a}]
Bisimilarity has no finite, complete axiomatisation over CCS.
\end{theorem}

To prove this result, in \cite{Mo90a}, Moller considered the following family of equations $\{M_n\}_{n \ge 1}$:
\begin{align*}
& (x + y) \mathbin{\|} \sum_{i = 1}^n z_i +
\sum_{i = 1}^n \Big(x \mathbin{\|} z_i + y \mathbin{\|} z_i \Big) \approx \\
& x \mathbin{\|} \sum_{i = 1}^n z_i +
y \mathbin{\|} \sum_{i = 1}^n z_i +
\sum_{i = 1}^n \Big((x + y) \mathbin{\|} z_i \Big)
\tag{$M_n$}\label{eq:moller_family}
\end{align*}
and he argued that all these equations should be sound modulo any behavioural congruence that is \emph{reasonable}, including bisimilarity. 
Roughly speaking, for each $n \ge 1$, the terms in the two sides of~\ref{eq:moller_family} can match exactly their single step behaviour and, at the same time, the equation does not introduce any causal dependency between the behaviour of the single components of each term.
Moller then considered a particular family of instantiations $\{I_n\}_{n \ge 1}$ of $\{M_n\}_{n \ge 1}$, consisting only of closed terms:
\begin{align*}
& (a + aa) \mathbin{\|} \sum_{i = 1}^n a^i +
\sum_{i = 1}^n \Big(a \mathbin{\|} a^i + aa \mathbin{\|} a^i \Big) 
\approx \\
& a \mathbin{\|} \sum_{i = 1}^n a^i + 
aa \mathbin{\|} \sum_{i = 1}^n a^i +
\sum_{i = 1}^n \Big((a + aa) \mathbin{\|} a^i \Big)
\tag{$I_n$}
\end{align*}
and he argued that no finite set of equations, that are sound modulo bisimilarity, can derive $I_n$ for each $n \ge 1$.

To this end, he applied the following proof strategy, which has been later referred to as the \emph{proof-theoretic} approach to negative results.
Whenever an equation $t \approx u$ is provable from an axiom system $\E$, then there is a \emph{proof} of it, i.e., a sequence of equations $t_i \approx u_i$, for $i = 1,\dots,n$, such that $t = t_n$, $u=u_n$, and each equation $t_i \approx u_i$ is in turn derivable from $\E \cup \{t_j \approx u_j \mid j < i\}$. 
The aim in the proof-theoretic approach is to show that no such sequence exists, so that the considered equation cannot be derived from $\E$.
To this end, we need to identify a specific property of terms, say $P_n$ for $n \ge 0$, that, when $n$ is \emph{large enough}, is preserved by provability from finite, sound axiom systems.
Roughly, this means that if:
\begin{itemize}
\item $\E$ is a finite set of axioms that are sound modulo $\sim$, 
\item the equation $p \approx q$ is provable from $\E$, and 
\item $n > \size(t)$ for any term $t$ in the equations in $\E$,
\end{itemize} 
then either both $p$ and $q$ satisfy $P_n$, or none of them does. 
Then, we exhibit an infinite family of sound equations in which $P_n$ is not preserved, namely it is satisfied only by one side of each equation.

Using this method, Moller proved that whenever $n$ is larger than the size of any term occurring in the equations in a finite, sound, axiom system $\E$, then the instance $I_n$ cannot be derived from $\E$.
(In this case, the property $P_n$ was to have a summand bisimilar to $(a + aa) \mathbin{\|} \sum_{i = 1}^n a^i$.)

All negative results that will be discussed in this paper have been obtained by applying the proof-theoretic approach.
Others examples of an application can be found in, e.g., \cite{AFIL05,AACIP19,AACILP20}.
We refer the interested reader to the survey paper \cite{AFIL05a} for a presentation of other existing techniques for proving negative (as well as positive) results in equational logic.


\section{Are two binary auxiliary operators necessary?}
\label{sec:no_binary}

The aforementioned positive results by Bergstra and Klop, and Moller's negative result, raise the following question:
\begin{equation}
\tag{Q1}\label{eq:problem}
\parbox{\dimexpr\linewidth-4em}{
\strut
\emph{Can we obtain a finite axiomatisation of parallel composition in bisimulation semantics by adding} 
only one binary operator $f$
\emph{to the signature of CCS?}
\strut
}
\end{equation}
The interest in (\ref{eq:problem}) is threefold, as an answer to it would:
\begin{enumerate}
\item provide the first study on the finite axiomatisability of operators whose operational semantics is not predetermined,
\item clarify the status of the auxiliary operators \emph{left merge} and \emph{communication merge}, proposed in \cite{BK84b}, in the finite axiomatisation of parallel composition, and
\item give further insight into properties that auxiliary operators used in the finite equational characterisation of parallel composition ought to afford.
\end{enumerate}
In \cite{ACFIL21} we prove that, under some reasonable simplifying assumptions, there is no auxiliary binary operator that can be added to CCS to yield a finite equational axiomatisation of bisimilarity. 
To the best of our knowledge, \cite{ACFIL21} proposes the first non-finite axiomatisability result for a process algebra in which one of the operators, namely the auxiliary operator $f$, does not have a fixed semantics, and it is thus a substantial generalisation of previous non-finite axiomatisability theorems by Moller~\cite{Mo89,Mo90} and Aceto et al.~\cite{AFIL05}.


\subsection{The simplifying assumptions}

In \cite{ACFIL21} we analyse the axiomatisability of parallel composition over the language $\FCCS$, i.e., CCS enriched with a binary operator $f$ that we use to express $\|$ as a derived operator:
\[
t ::=\; \nil \;|\; 
x \;|\; 
a.t \;|\; 
\bar{a}.t \;|\; 
\tau. t \;|\; 
t+t \;|\; 
t \mathbin{\|} t \;|\;
f(t_1,t_2)
\enspace ,
\]
where the set of actions is $\Acttau = \{a,\overline{a},\tau\}$.

We prove that an auxiliary operator $f$ alone does not allow us to obtain a finite ground-complete axiomatisation of $\FCCS$ modulo bisimilarity.
However, for our technical developments, it has been necessary to restrict the search space for $f$ to a set of \emph{meaningful} operators, by means of the aforementioned simplifying assumptions.
Clearly, $f$ can be considered meaningful for a solution of (\ref{eq:problem}) only if it is an operator that preserves bisimilarity. 
Hence, the first assumption we make on the auxiliary operator $f$ guarantees that it meets such requirement.

One way to guarantee that $f$ preserves bisimilarity is to postulate that the behaviour of $f$ is described using Plotkin-style rules that fit a rule format that is known to preserve bisimilarity, see, e.g.,~\cite{AFV01} for a survey of such rule formats. 
The simplest format satisfying this criterion is the format proposed by de Simone \cite{dS85}. 
Hence, we make the following

\begin{assumption}
\label{Ass:deSimone}
The behaviour of $f$ is described by rules in de Simone format. 
\end{assumption}

Notice that being formally defined by rules in de Simone format is the only knowledge we assume on the operational semantics of $f$.

Our second simplifying assumption concerns how the operator $f$ can be used to axiomatise parallel composition.
We postulate that the behaviour of the parallel composition operator is expressed equationally by a law that is akin to the one used by Bergstra and Klop to define $\mathbin{\|}$ in terms of $\lmerge$ and $\cmerge$. 
Formally, we express this postulate as the following
\begin{assumption}
\label{Ass:equation}
The equation
\begin{equation}
\label{eq:intro}
x \mathbin{\|} y \approx f(x,y) + f(y,x) 
\tag{PF}
\end{equation}
is sound modulo bisimilarity.
\end{assumption}

We then proceed by a case analysis over the possible sets of de Simone rules defining the behaviour of $f$, in such a way that the validity of Equation~\eqref{eq:intro} modulo bisimilarity is guaranteed.
To fully characterise the sets of rules that may define $f$, we introduce a third simplifying assumption: 
\begin{assumption}
\label{Ass:targets}
The target of each rule for $f$ is either a variable or a term obtained by applying a single $\FCCS$ operator to the variables of the rule, according to the constraints of the de Simone format.
\end{assumption}
Then, for each of the resulting cases, we show the desired negative result using the proof-theoretic technique described in Section~\ref{sec:classic}.

To our mind, those three assumptions are `reasonable' because they allow us to simplify the combinatorial complexity of our analysis without excessively narrowing down the set of operators captured by our approach.
In particular, there are three main reasons behind the assumption on the rules for $f$ being in de Simone format:
\begin{itemize}
\item The de Simone format is the simplest congruence format for bisimilarity.
Hence we must be able to deal with this case before proceeding to any generalisation.
\item The specification of parallel composition, left merge and communication merge operators (and of the vast majority of process algebraic operators) is in de Simone format.
Hence, that format was a natural choice also for $f$.
\item The simplicity of the de Simone rules allows us to reduce considerably the complexity of our case analysis over the sets of available rules for $f$.
However, as witnessed by the developments in \cite{ACFIL21}, even with this simplification, the proof of the desired negative result requires a large amount of delicate, technical work.
\end{itemize}
The other two assumptions still allow us to obtain a significant generalisation of related works, such as \cite{AFIL05}, as they are an attempt to identify the requirements needed to apply Moller's proof technique to Hennessy's-merge-like operators.


\subsection{The results in a nutshell}

The solution of~\eqref{eq:problem} boils down to proving the following:

\begin{theorem}
\label{Thm:nonfin-en}
Suppose that our three assumptions are met.
Let $\E$ be a finite axiom system over $\FCCS$ that is sound modulo bisimilarity. 
Then there is an infinite family $\{e_n\}_{n\geq 0}$ of sound equations such that $\E$ does not prove equation $e_n$ for each $n$ larger than the size of each term in the equations in $\E$. 
\end{theorem}

The proof of the above result presented in \cite{ACFIL21} is built on four steps, that we proceed to sketch in the remainder of this section.

As a first step, we use the simplifying assumptions to map out the universe of \emph{meaningful} operators $f$.
In particular, by means of the following proposition we can characterise the operational semantics of each meaningful operator $f$.

\begin{proposition}
\label{prop:f_rules}
If an operator $f$ meets our three simplifying assumptions, then $f$ must have a rule of the form 
\[
\SOSrule{x \trans[\alpha] x' \quad y \trans[\overline{\alpha}] y'}{f(x,y) \trans[\tau] t_1(x',y')}
\]
for at least one $\alpha \in \{a,\overline{a}\}$, and at least one rule of the form
\[
\SOSrule{x \trans[\mu] x'}{f(x,y) \trans[\mu] t_2(x',y)}
\quad\text{ or }\quad
\SOSrule{y \trans[\mu] y'}{f(x,y) \trans[\mu] t_3(x,y')}
\]
for each $\mu\in\Acttau$, where $t_1,t_2,t_3$ are such that $t_{i}(x,y) \sim_\B x \mathbin{\|} y$ for each $i \in \{1,2,3\}$.
\end{proposition}

Notice that when the set of actions has $2n+1$ elements, there are $3^{3n+1}$ possible operators meeting those constraints.

In the second step, we show that no meaningful operator $f$ distributes over $+$ in both arguments.
We then split the family of meaningful operators into two sub-families: 
one consisting of operators that distribute over $+$ in one argument, and the other containing those that do not distribute over $+$.
The reason behind this splitting is purely technical: although the general strategy we use to obtain the negative result for the two sub-families is the same, the are some preliminary technical results that hold for only one of the two families, thus forcing us to treat the two cases separately.
(All the details can be found in the technical report version \cite{ACFIL20} of \cite{ACFIL21}.)

Then, as third step, we prove the negative result for all meaningful operators $f$ that distribute over $+$ in one argument.

\begin{example}
Assume that the set of SOS rules for $f$ is the following: 
\begin{equation}
\label{eq:rules_ex1}
\SOSrule{x \trans[\mu] x'}{f(x,y) \trans[\mu] x' \mathbin{\|} y}\; \forall\,\mu\in\Act
\qquad
\SOSrule{x \trans[\alpha] x' \quad y \trans[\bar{\alpha}] y'}{f(x,y) \trans[\tau] x' \mathbin{\|} y'}
\end{equation}
According to the proof-theoretic approach sketched in Section~\ref{sec:classic}, we now introduce a particular family of equations on which we will build our negative result.
For each action $\mu \in \Acttau$ and a non-negative integer $i$, we let $\mu^{\scriptstyle \leq i} = \mu + \mu^2 + \cdots + \mu^i$.
Then we define
\begin{align*}
& p_n  =  \sum_{i=0}^{n} \bar{\alpha} \alpha^{\scriptstyle \leq i} & (n \ge 0) \enspace \phantom{.} \\
& \epsilon_n \colon \quad f(\alpha,p_n)  \approx  \alpha p_n + \sum_{i=0}^{n} \tau \alpha^{\scriptstyle \leq i} & (n\ge 0) \enspace,
\end{align*}
As formalised in the following theorem, having a summand bisimilar to $f(\alpha,p_n)$, for `large enough $n$', is a specific property of terms that is satisfied by all instatiations of axioms in $\E$ and it is preserved by the rules of equational logic.

\begin{theorem}
\label{thm:Labat}
Assume an operator $f$ defined by the inference rules in Equation~\eqref{eq:rules_ex1}.
Let $\E$ be a finite axiom system over $\FCCS$ that is sound modulo $\sim_\B$, 
$n$ be larger than the size of each term in the equations in $\E$, and $p,q$ be closed terms such that $p,q \sim_\B f(\alpha,p_n)$. 
If $\E \vdash p \approx q$ and $p$ has a summand bisimilar to $f(\alpha,p_n)$, then so does $q$.
\end{theorem}

Then, since the left-hand side of equation $\epsilon_n$, viz.~the term $f(\alpha,p_n)$, has a summand bisimilar to $f(\alpha,p_n)$, whilst the right-hand side, viz.~the term $\alpha p_n + \sum_{i=0}^{n} \tau \alpha^{\scriptstyle \le i}$, does not, we can conclude that the infinite collection of equations $\{\epsilon_n\}_{n \ge 0}$ is the desired witness family.
Theorem~\ref{Thm:nonfin-en} is then proved for the class of auxiliary binary operators defined by the inference rules given above.
\end{example}

The fourth, and final, step consists in proving the negative result for all meaningful operators $f$ that do not distribute over $+$.

\begin{example}
We choose $\alpha \in \{a,\bar{a}\}$ and we assume that the set of rules for $f$ includes 
\begin{equation}
\label{eq:rules_ex2}
\SOSrule{x \trans[\alpha] x'}{f(x,y) \trans[\alpha] x' \mathbin{\|} y}
\qquad
\SOSrule{y \trans[\alpha] y'}{f(x,y) \trans[\alpha] x \mathbin{\|} y'}
\end{equation}
We now introduce the infinite family of valid equations, modulo bisimilarity, that will allow us to obtain the negative result in the case at hand.
We define
\begin{align*}
& q_n = \sum_{i = 0}^n \alpha\bar{\alpha}^{\scriptsize \le i} & (n \ge 0) \enspace\phantom{.} \\
& \varepsilon_n \colon \quad f(\alpha,q_n) \approx \alpha q_n + \sum_{i = 0}^n \alpha ( \alpha \mathbin{\|} \bar{\alpha}^{\scriptsize \le i} ) & (n \ge 0) \enspace.
\end{align*}

Following the proof-theoretic approach, we aim to show that, when $n$ is \emph{large enough}, the term specific property of having a summand bisimilar to $f(\alpha,q_n)$ is preserved by derivations from a finite, sound axiom system $\E$, as stated in the following theorem:

\begin{theorem}
\label{thm:LaRa}
Assume an operator $f$ defined by the inference rules in Equation~\eqref{eq:rules_ex2}.
Let $\E$ be a finite axiom system over $\FCCS$ that is sound modulo $\sim_\B$, $n$ be larger than the size of each term in the equations in $\E$, and $p,q$ be closed terms such that $p,q \sim_\B f(\alpha,q_n)$.
If $\E \vdash p \approx q$ and $p$ has a summand bisimilar to $f(\alpha,q_n)$, then so does $q$.
\end{theorem}

Then, we can conclude that the infinite collection of equations $\{\varepsilon_n\}_{n \ge 0}$ is the desired family showing the negative result: 
The left-hand side of equation $\varepsilon_n$, viz.~the term $f(\alpha,q_n)$, has a summand bisimilar to $f(\alpha,q_n)$, whilst the right-hand side, viz.~the term $\alpha q_n + \sum_{i=0}^{n} \alpha(\alpha \mathbin{\|} \bar{\alpha}^{\scriptstyle \le i})$, does not.
\end{example}


\section{Beyond bisimilarity: A journey in the spectrum}
\label{sec:spectrum}

\begin{figure*}[t]
\centering
\scalebox{0.9}{
\begin{tikzpicture}
\node at (0,7){bisimulation ($\sim_\B$)}; 
\draw[-latex](0,6.8)--(0,6.2);
\node at (0,6){\textcolor{red}{$2$-nested simulation ($\sim_{2\s}$)}};
\draw[-latex](0,5.8)--(0,5.2);
\node at (-1.9,5){\textcolor{blue}{failure simulation ($\sim_{\fail\s}$) $=$ ready simulation ($\sim_{\rs}$)}};
\draw[-latex](0,4.8)--(0,4.2);
\node at (0,4){\textcolor{blue}{ready trace ($\sim_\rtr$)}};
\draw[-latex](0,3.8)--(-2,3.2);
\draw[-latex](0,3.8)--(2,3.2);
\node at (-2,3){\textcolor{blue}{failure trace ($\sim_\ftr$)}};
\node at (2,3){\textcolor{blue}{readies ($\sim_\ready$)}};
\draw[-latex](-2,2.8)--(-0.2,2.2);
\draw[-latex](2,2.8)--(0,2.2);
\node at (0,2){\textcolor{blue}{failures ($\sim_\fail$)}};
\draw[-latex](0,1.8)--(0,1.2);
\node at (0,1){\textcolor{blue}{completed trace ($\sim_\ctr$)}};
\draw[-latex](0,0.8)--(0,0.2);
\node at (0,0){\textcolor{blue}{trace ($\sim_\tr$)}};
\draw[-latex](-0.1,4.8)--(-6.5,3.2);
\node at (-6.5,3){\textcolor{blue}{completed simulation ($\sim_{\mathtt{C}\s}$)}};
\draw[-latex](-6.5,2.8)--(-6.5,2.2);
\draw[-latex](-6.3,2.8)--(-0.2,1.2);
\node at (-6.5,2){\textcolor{blue}{simulation ($\sim_\s$)}};
\draw[-latex](-6.5,1.8)--(-0.2,0.2);
\draw[-latex](0.1,5.8)--(5.3,4.7);
\node at (5.3,4.5){\textcolor{red}{possible futures ($\sim_\pf$)}};
\draw[-latex](5.3,4.3)--(2.2,3.2);
\draw[dashed,ultra thick,ForestGreen](-5,5.7)--(1.7,5.4);
\draw[dashed,ultra thick,ForestGreen](1.7,5.4)--(2.5,5);
\draw[dashed,ultra thick,ForestGreen](2.5,5)--(4,1);
\end{tikzpicture}
}
\caption{The linear time-branching time spectrum \cite{vG90,vG01}.
The dashed (\textcolor{ForestGreen}{green}) line delineates the boundary between finite and non-finite axiomatisability of the behavioural congruences: 
For the ones below the boundary (in \textcolor{blue}{blue}) we provide a finite, ground-complete axiomatization; for the ones above the boundary (in \textcolor{red}{red}), we provide a negative result.
The case of bisimulation follows from \cite{Mo89,Mo90,Mo90a}.
\label{fig:spectrum}
}
\end{figure*}
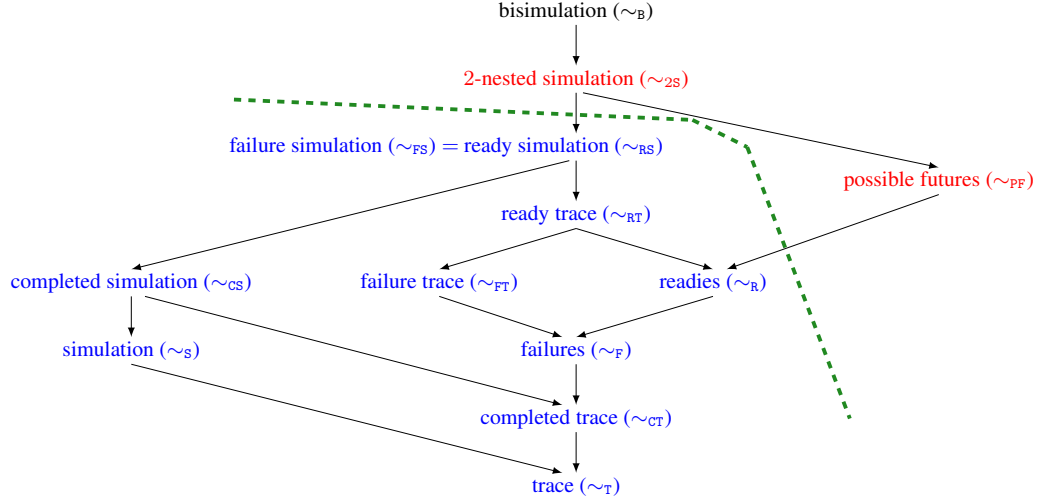

So far we have discussed the axiomatisability of the parallel composition operator modulo bisimilarity.
In the light of the plethora of other semantics in the spectrum, it is natural to address the following question: 
\begin{equation}
\tag{Q2}\label{eq:spectrum}
\parbox{\dimexpr\linewidth-4em}{
\strut
\emph{Can we obtain a finite axiomatisation of parallel composition, over CCS,} modulo congruences other \emph{than bisimilarity?}
\strut
}
\end{equation}
In our recent work \cite{ACILP20} we successfully delineated the \emph{boundary} between finite and non-finite axiomatisability of the congruences in the spectrum over the language CCS, as shown in Figure~\ref{fig:spectrum}.
(We remark that \cite{ACILP20} deals with the pure interleaving parallel composition operator. 
The results we are going to discuss on the full merge operator can be found in \cite{ACILP21}.)


\subsection{The positive part of the spectrum}

We start by briefly discussing the \emph{finite}, \emph{ground-complete} axiomatisation for \emph{ready simulation} semantics \cite{BIM95} obtained in \cite{ACILP20,ACILP21}.

\begin{definition}
[Ready simulation equivalence]
\label{def:rs}
Let $(\Proc,\Acttau,\trans[])$ be a LTS.
A \emph{simulation} is a binary relation ${\rel} \subseteq \Proc \times \Proc$ such that, whenever $p \rel q$ and $p \trans[a] p'$, then there is some $q'$ such that $q \trans[a] q'$ and $p' \rel q'$.
A \emph{ready simulation} is a simulation $\!\rel\!$ such that, whenever $p \rel q$ then $\init(p) = \init(q)$.
We write $p \sqsubseteq_\rs q$ if there is a ready simulation $\!\rel\!$ such that $p \rel q$.
We say that $p$ is \emph{ready simulation equivalent} to $q$, notation $p \sim_\rs q$, if $p \sqsubseteq_\rs q$ and $q \sqsubseteq_\rs p$.
\end{definition}

\begin{table*}[t]
\[
\begin{array}{ll}
\scalebox{0.95}{(RS)} & \mu (\nu x + \nu y + z) \approx \mu (\nu x + \nu y + z) + \mu (\nu x + z) \\[.2cm]
\scalebox{0.95}{(RSP1)} & (\mu x + \mu y + u) \mathbin{\|} (\nu z + \nu w + v) \approx (\mu x + u) \mathbin{\|} (\nu z+ \nu w + v) + + (\mu  y + u) \mathbin{\|} (\nu z + \nu w + v) + \\[.1cm]
& \phantom{(\mu x + \mu y + u) \mathbin{\|} (\nu z + \nu w + v) \approx}
+ (\mu x + \mu y + u) \mathbin{\|} (\nu z + v) + (\mu x + \mu y + u) \mathbin{\|} (\nu w + v) \\[.2cm]
\scalebox{0.95}{(RSP2)} & \displaystyle \Big(\sum_{i \in I} \mu_i x_i \Big) \mathbin{\|} (\nu y + \nu z + w) \approx \Big(\sum_{i \in I} \mu_i x_i \Big) \mathbin{\|} (\nu y + w) + \Big(\sum_{i \in I} \mu_i x_i \Big) \mathbin{\|} (\nu z + w) +  \\[.1cm]
& \displaystyle \phantom{\Big(\sum_{i \in I} \mu_i x_i \Big) \mathbin{\|} (\nu y + \nu z + w) \approx}
+ \sum_{i \in I} \mu_i \big(x_i \mathbin{\|} (\nu y + \nu z + w) \big) \\[.1cm]
& \text{where } \mu_j \neq \mu_k \text{ if } j \neq k \text{ for } j,k \in I \\[.2cm]
& \hspace{-1cm}\E_\rs = \E_1 \cup \{\textrm{RS, RSP1, RSP2, EL2}\}
\\[.2cm]
\end{array}
\]
\caption{Additional axioms for ready simulation equivalence.}
\label{tab:ready_sim_axioms}
\end{table*}

In \cite{vG90} it was proved that the axiom system consisting of $\E_0$ (Table~\ref{tab:basic_axioms}) together with axiom RS in Table~\ref{tab:ready_sim_axioms} is a ground-complete axiomatisation of BCCSP, i.e., the language that is obtained from CCS if $\mathbin{\|}$ is omitted, modulo $\sim_\rs$.
Hence, to obtain a finite, ground-complete axiomatisation of CCS modulo $\sim_\rs$ it suffices to enrich the axiom system $\E_1 \cup \{\mathrm{RS}\}$ with finitely many axioms allowing one to eliminate all occurrences of $\mathbin{\|}$ from closed CCS terms, that is to prove that for every closed CCS term $p$ there is a closed BCCSP term $q$ such that $\E_\rs \vdash p \approx q$.
Then, the completeness of the proposed axiom system over CCS is a direct consequence of that over BCCSP proved in \cite{vG90}.

Clearly, EL would allow us to obtain the desired elimination, but, as previously mentioned, it is a schema that finitely presents an infinite collection of equations, and thus an axiom system including it is infinite.
However, thanks to the schemata RSP1 and RSP2 that characterise the distributivity of $\mathbin{\|}$ over $+$ modulo $\sim_\rs$ (see Table~\ref{tab:ready_sim_axioms}), we can include EL2 instead, which is a variant of EL that generates only finitely many axioms (see Table~\ref{tab:exp_law}).

Since the axioms for the elimination of parallel composition modulo ready simulation equivalence are, of course, sound with respect to equivalences that are coarser than ready simulation equivalence, the `reduction to ground-completeness over BCCSP' works for all behavioural equivalences in the spectrum below ready simulation equivalence. 
Nevertheless, for those equivalences, we offer more elegant axioms to equationally eliminate parallel composition from closed terms, as reported in Table~\ref{tab:spectrum_axioms}.
We also observe a sort of parallelism between the axiomatisations for the notions of simulation and the corresponding decorated trace semantics: the axioms used to equationally express the interplay between the interleaving operator and the other operators of BCCSP in a decorated trace semantics can be seen as the \emph{linear counterpart} of those used in the corresponding notion of simulation semantics.
For instance, while the axioms for ready simulation impose constraints on the form of both arguments of the interleaving operator to facilitate equational reductions, those for ready trace equivalence impose similar constraints but only on one argument.

\begin{table*}
\centering
\begin{tabular}{l}
\scalebox{0.95}{(CS)} \, $\mu ( \nu x + y + z) \approx \mu (\nu x + y + z) + \mu (\nu x + z)$ \\[.2cm]
\scalebox{0.95}{(CSP1)} $(\mu x + \nu y + u) \mathbin{\|} (\xi z+ \zeta w + v) \approx (\mu x + u) \mathbin{\|} (\xi z + \zeta w + v) + (\nu y + u) \mathbin{\|} (\xi z + \zeta w + v) +$\\[.1cm]
\qquad\quad\phantom{$(\mu x + \nu y + u) \mathbin{\|} (\xi z + \zeta w + v) \approx$}
$+ (\mu x + \nu y + u) \mathbin{\|} (\xi z + v) + (\mu x + \nu y + u) \mathbin{\|} (\zeta w + v)$ \\[.2cm]
\scalebox{0.95}{(CSP2)} $\mu x \mathbin{\|} (\nu y + \xi z + w) \approx \mu (x \mathbin{\|} (\nu y + \xi z + w)) + \mu x \mathbin{\|} (\nu y + w) + \mu x \mathbin{\|} (\xi z + w)$ \\[.2cm]
\; $\E_\cs = \E_1 \cup \{\textrm{CS, CSP1, CSP2, EL1}\}$
\\[.1cm]
\hline 
\hline 
\\
\scalebox{0.95}{(S)} \, $\mu (x+y) \approx \mu(x+y) + \mu x$ \\[.2cm]
\scalebox{0.95}{(SP1)} \, $(x+y) \mathbin{\|} (z+w) \approx x \mathbin{\|} (z+w) + y \mathbin{\|} (z+w) + (x+y) \mathbin{\|} z + (x+y) \mathbin{\|} w$ \\[.2cm]
\scalebox{0.95}{(SP2)} \, $\mu x \mathbin{\|} (y+z) \approx \mu (x \mathbin{\|} (y+z)) + \mu x \mathbin{\|} y + \mu x \mathbin{\|} z$ \\[.2cm]
\; $\E_\s = \E_1 \cup \{\textrm{S, SP1, SP2, EL1}\}$ \\[.05cm]
\hline 
\hline 
\\
\, \scalebox{0.95}{(RT)} \; $\mu \left(\sum_{i=1}^{|\Act|}(\nu_i x_i + \nu_i y_i) + z\right) \approx \mu \left(\sum_{i=1}^{|\Act|}\nu_i x_i + z\right) + \mu \left(\sum_{i=1}^{|\Act|}\nu_i y_i + z\right)$ \\[.2cm]
\, \scalebox{0.95}{(FP)} \; $(\mu x + \mu y + w) \mathbin{\|} z \approx (\mu x + w) \mathbin{\|} z + (\mu y + w) \mathbin{\|} z$ \\[.2cm]
\; $\E_\rtr = \E_1 \cup \{\textrm{RT}, \textrm{FP}, \textrm{EL2}\}$
\\[.1cm]
\hline
\hline
\\
\, \scalebox{0.95}{(FT)} \; $\mu x + \mu y \approx \mu x + \mu y + \mu(x + y)$ \\[.2cm]
\; $\E_\ftr = \E_1 \cup \{\textrm{FT},\textrm{RS},\textrm{FP},\textrm{EL2}\}$
\\[.1cm]
\hline
\hline
\\
\, \scalebox{0.95}{(R)} \; $\mu( \nu x + z) + \mu (\nu y + w) \approx \mu (\nu x + \nu y + z) + \mu (\nu y + w)$ \\[.2cm]
\; $\E_\ready = \E_1 \cup \{\textrm{R, FP, EL2}\}$
\\[.1cm]
\hline
\hline
\\
\, \scalebox{0.95}{(F)} \; $\mu x + \mu(y + z) \approx \mu x + \mu(x + y) + \mu(y + z)$ \\[.2cm]
\; $\E_\fail = \E_1 \cup \{\textrm{F, R, FP, EL2}\}$
\\[.1cm]
\hline
\hline
\\
\, \scalebox{0.95}{(CT)} \; $\mu (\nu x + z) + \mu (\xi y + w) \approx \mu (\nu x + \xi y + z + w)$ \\[.2cm]
\, \scalebox{0.95}{(CTP)} \; $(\mu x + \nu y + w) \mathbin{\|} z \approx (\mu x + w) \mathbin{\|} z + (\nu y + w) \mathbin{\|} z$ \\[.2cm]
\; $\E_\ctr = \E_1 \cup\{\textrm{CT, CTP, EL1}\}$
\\[.1cm]
\hline
\hline
\\
\, \scalebox{0.95}{(T)} \; $\mu x + \mu y \approx \mu (x + y)$ \\[.2cm]
\, \scalebox{0.95}{(TP)} \; $(x + y) \mathbin{\|} z \approx x \mathbin{\|} z + y \mathbin{\|} z$ \\[.2cm]
\; $\E_\tr = \E_1 \cup\{\textrm{T, TP, EL1}\}$
\\[.2cm]
\end{tabular}
\caption{Additional axioms for the semantics in the spectrum that are coarser than ready simulation.
}
\label{tab:spectrum_axioms}
\end{table*}

Our results can be then summarised as follows:

\begin{theorem}
\label{thm:spectrum_positive}
Let $\mathtt{X} \in \{\rs,\cs,\s,\rtr,\ftr,\ready,\fail,\ctr,\tr\}$.
The axiom system $\E_{\mathtt{X}}$ is sound and ground-complete modulo $\sim_{\mathtt{X}}$ over CCS.
\end{theorem}


\subsection{The negative part of the spectrum}

We complete our journey in the spectrum by showing that \emph{nested simulation} and \emph{nested trace} semantics do not have a finite axiomatisation over CCS.
To this end, we adapt Moller's arguments to the effect that bisimilarity is not finitely based over CCS to obtain the negative result for \emph{possible futures equivalence} \cite{RB81}, also known as $2$-\emph{nested trace equivalence}.

\begin{definition}
[Possible futures equivalence]
\label{def:pf}
A \emph{possible future} of a process $p$ is a pair $(\varphi,X)$ where $\varphi \in \Acttau^{*}$ and $X \subseteq \Acttau$ such that $p \trans[\varphi] p'$ for some $p'$ and $X$ is the set of traces of $p'$.
We write $\pf(p)$ for the set of possible futures of $p$.
Two processes $p$ and $q$ are said to be \emph{possible futures equivalent}, denoted $p \sim_\pf q$, if $\pf(p) = \pf(q)$.
\end{definition}

Consider the infinite family of equations $\{e_N\}_{N \ge 1}$ given, for $a \neq b$, by:
\begin{align*}
& p_N = \sum_{i = 1}^N b^i a 
& (N \ge 1) \enspace \phantom{.} \\
& e_N \; \colon \; a \mathbin{\|} p_N \approx a p_N + \sum_{i = 1}^{N} b (a \mathbin{\|} b^{i-1}a)   
& (N \ge 1) \enspace .
\end{align*}

Notice that the equations $e_N$ are sound modulo possible futures equivalence for all $N \ge 1$.

We also notice that none of the summands in the right-hand side of equation $e_N$ is, alone, possible futures equivalent to $a \mathbin{\|} p_N$.
However, we now proceed to show that, when $N$ is large enough, having a summand possible futures equivalent to $a \mathbin{\|} p_N$ is an invariant under provability from finite sound axiom systems.
Hence, the negative result for possible futures equivalence is a consequence of the following theorem:

\begin{theorem}
\label{thm:pf_axiom_derivation}
Let $\E$ be a finite axiom system over CCS that is sound modulo $\sim_\pf$.
Let $N$ be larger than the size of each term in the equations in $\E$.
Assume that processes $p$ and $q$ are such that $p,q \sim_\pf a \mathbin{\|} p_N$.
If $\E \vdash p \approx q$ and $p$ has a summand possible futures equivalent to $a \mathbin{\|} p_N$, then so does $q$.
\end{theorem}

Then, we exploit the soundness modulo bisimilarity of the equations in the family $\{e_N\}_{N \ge 1}$ to extend the negative result to all the congruences that are finer than possible futures and coarser than bisimilarity, thus including all nested trace and nested simulation semantics.

\begin{theorem}
\label{thm:spectrum_negative}
Assume that ${|\Act|}\ge 2$.
Let $n \ge 2$.
Then, $n$-nested trace equivalence and $n$-nested simulation equivalence admit no finite, ground-complete, equational axiomatisation over the language CCS.
\end{theorem}


\section{Introducing silent steps}
\label{sec:weak}

In the previous sections we have considered the $\tau$ action as an observable move by a process.
We now switch from strong to weak semantics: in this new setting, a $\tau$-move corresponds to a \emph{silent} (or \emph{hidden}, \emph{invisible}) step in the behaviour of a process.
In detail, we are interested in studying an equational characterisation of the parallel composition operator modulo \emph{rooted weak bisimilarity}.

Let $\trans[\varepsilon]$ denote the reflexive and transitive closure of the transition $\trans[\tau]$. 
Then, let $p \xRightarrow{\,\mu\,} q$ be a shorthand for $p \trans[\varepsilon] \trans[\mu] \trans[\varepsilon] q$, if $\mu \neq \tau$, or $p \trans[\varepsilon] q$ if $\mu = \tau$.
Finally, let $p \xRightarrow{\,\hat{\mu}\,} q$ be a shorthand for $p \trans[\varepsilon] \trans[\mu] \trans[\varepsilon] q$, for all $\mu \in \Acttau$ (notice that $\xRightarrow{\,\hat{\mu}\,}$ differs from $\xRightarrow{\,\mu\,}$ only when $\mu = \tau$).

\begin{definition}
[Rooted weak bisimilarity]
\label{def:rwb}
Let $(\Proc,\Acttau,\trans[])$ be a LTS.
\emph{Weak bisimilarity}, denoted by $\sim_\wb$, is the largest binary symmetric relation over $\Proc$ such that whenever $p \sim_\wb q$ and $p \trans[\mu] p'$, then either
\begin{itemize}
\item $\mu = \tau$ and $p' \rel q$, or
\item there is a processes $q'$ such that $q \xRightarrow{\,\mu\,} q'$ and $p' \rel q'$.
\end{itemize}
Then, \emph{rooted weak bisimilarity}, denoted by $\sim_\rwb$, is the binary symmetric relation over $\Proc$ such that whenever $p \sim_\rwb q$ and $p \trans[\mu] p'$, then there is a process $q'$ such that $q \xRightarrow{\,\hat{\mu}\,} q'$ and $p' \sim_\wb q'$.
\end{definition}

It is well known that rooted weak bisimilarity is an equivalence relation, and that the root condition is necessary to guarantee the compositionality with respect to the nondeterministic choice operator (as well as the left merge), see, e.g., \cite{BK85,vGW96}, and thus that $\sim_\rwb$ is a congruence over CCS.

In this section we present the original contribution of our survey, namely a negative answer to the following problem:
\begin{equation}
\tag{Q3}\label{eq:weak}
\parbox{\dimexpr\linewidth-4em}{
\strut
\emph{Can we obtain a finite axiomatisation of parallel composition} modulo rooted weak bisimilarity \emph{over CCS}?
\strut
}
\end{equation}
Our aim is to prove the following theorem:

\begin{theorem}
\label{thm:rwb_not_fin}
Rooted weak bisimilarity has no finite, complete axiomatisation over CCS. 
\end{theorem}

To this end, we exploit the family of equations $\{M_n\}_{n \ge 1}$, from \cite{Mo90a}, introduced in Section~\ref{sec:classic}.
First of all, notice that the equations $M_n$ are all sound modulo rooted weak bisimilarity, as bisimilarity is included in rooted weak bisimilarity, i.e. $t \sim_\B u$ implies $t \sim_\rwb u$ for all CCS terms $t,u$.
Then, we remark that in \cite{Mo90a} Moller obtained his result over a fragment of the language CCS that have considered in the paper.
In particular, he considered the \emph{purely interleaving parallel composition} operator, i.e., parallel composition without communication.
Notice that if we restrict the set of actions from $\Acttau$ to $\Act$, then there is no difference between full parallel composition and interleaving, since the lack of actions co-names prevents any form of synchronisation between CCS terms.
This means that the CCS terms considered by Moller were built over the set of actions $\Act$, in place of $\Acttau$.
Let us denote by CCS$\!\!\phantom{.}_{\Act}$ the fragment of CCS considered by Moller in \cite{Mo90a}.

Informally, the core of our proof consists in showing that any equation over CCS that is sound modulo $\sim_\rwb$ and that does not contain any occurrence of the prefixing operator, is \emph{also} sound modulo $\sim_\B$ over CCS$\!\!\phantom{.}_{\Act}$.
Then we show that, since all the terms occurring in the family $\{M_n\}_{n \ge 1}$ do not contain any occurrence of prefixing, any proof of an equation~\ref{eq:moller_family} from an axiom system sound modulo $\sim_\rwb$, uses \emph{only} equations over terms that do not contain any occurrence of prefixing.
Consequently, any finite axiom system that is sound modulo $\sim_\rwb$ over CCS, and can prove all the equations in the family $\{M_n\}_{n \ge 1}$, would also be sound modulo $\sim_\B$ over CCS$\!\!\phantom{.}_{\Act}$. 
As this contradicts the negative result obtained by Moller in \cite{Mo90a}, we can conclude that rooted weak bisimilarity has no finite, complete axiomatisation over CCS.

We devote the remainder of this section to a formalisation of the intuitions given above.
We remark that, although we formally discuss only the case of $\sim_\rwb$, our negative result can be extended to any weak congruence $\sim$ such that: $M_n$ is sound modulo $\sim$ for all $n \ge 1$, $\sim$ coincides with $\sim_\B$ over CCS$\!\!\phantom{.}_{\Act}$, and whenever $p \sim q$ then any initial $\tau$-step by $p$ is matched by $q$ and viceversa.
In particular, our result holds for the \emph{rooted} versions of \emph{branching bisimilarity}, \emph{delay bisimilarity} and $\eta$-\emph{bisimilarity}.

Firstly, we introduce the notion of \emph{action-free} terms, i.e., CCS terms that do not contain any occurrence of prefixing.

\begin{definition}
[Action-free term]
\label{def:action_free}
Let $t$ be a CCS term.
We say that $t$ is \emph{action-free} if $t \ntrans[\mu]$ for all $\mu \in \Acttau$.

An equation $t \approx u$ is \emph{action-free} if $t$ and $u$ are action-free.
\end{definition}

A fundamental property of action-free equations is that their soundness modulo rooted weak bisimilarity over CCS implies soundness modulo bisimilarity over CCS$\!\!\phantom{.}_{\Act}$.

\begin{proposition} 
\label{prop:act_free_rwb_b}
Let $t,u$ be action-free CCS terms.
If $t \approx u$ is sound modulo $\sim_\rwb$ over CCS, then it is also sound modulo $\sim_\B$ over CCS$\!\!\phantom{.}_{\Act}$.
\end{proposition}

\begin{proof}
Assume that $t \approx u$ is action-free and sound modulo $\sim_\rwb$. 
Let $\sigma$ be any closed substitution mapping variables to a CCS$\!\!\phantom{.}_{\Act}$ processes.
We remark that processes in CCS$\!\!\phantom{.}_{\Act}$ do not contain any occurrence of action $\tau$. 
By the soundness of $t \approx u$, we have that $\sigma(t) \sim_\rwb \sigma(u)$. Since $\sim_\B$ coincides with $\sim_\rwb$ over $\tau$-free processes, we obtain that $\sigma(t) \sim_\B \sigma(u)$. 
Hence, by the arbitrariness of $\sigma$, we can conclude that $\sigma_{\Act}(t) \sim_\B \sigma_{\Act}(u)$ for all closed substitutions $\sigma_{\Act}$ over CCS$\!\!\phantom{.}_{\Act}$, thus giving that $t \sim_\B u$ over CCS$\!\!\phantom{.}_{\Act}$.
\end{proof}

We identify a particular substitution, denoted by $\sigma_\nil$, that maps each variable to the null process.
Formally, the substitution $\sigma_\nil$ is defined as $\sigma_\nil(x) = \nil$, for all $ x \in \Var$.

\begin{lemma}
\label{lem:emptysubstitutionAct} 
Let $t$ be a CCS term. 
\begin{enumerate}
\item \label{lem:act_free}
If $t$ is action-free, then $\sigma_\nil(t) \sim_\B \nil$.
\item \label{lem:not_act_free}
If $t$ is not action-free, then there exists an action $\mu \in \Acttau$ such that $\sigma(t) \trans[\mu]$, for any substitution $\sigma$.
\end{enumerate}
\end{lemma} 

\begin{proof}
In both cases, the proof follows by induction over the structure of the term $t$.
\end{proof}

We now proceed to show that proofs of action-free equations, from an axiom system that is sound modulo rooted weak bisimilarity, use only action-free equations. 

\begin{proposition}
\label{prop:act_free_rwb}
Let $\E$ be an axiom system sound modulo $\sim_\rwb$.
\begin{enumerate}
\item \label{prop:both_act_free}
If $t \approx u$ is sound modulo $\sim_\rwb$ and $t$ is action-free, then also $u$ is action-free.
\item \label{prop:only_act_free}
If $\E \vdash t \approx u$ and $t$ is action-free, then a proof of $t \approx u$ from $\E$ uses only action-free equations. 
\end{enumerate}
\end{proposition}

\begin{proof}
We start from the first item.
As $t \approx u$ is sound modulo $\sim_\rwb$, we get that $\sigma_\nil(t) \sim_\rwb \sigma_\nil(u)$. 
Moreover, as $t$ is action-free, by Lemma~\ref{lem:emptysubstitutionAct}.\ref{lem:act_free} we have that $\sigma_\nil(t)$ cannot perform any action. 
Hence, $\sigma_\nil(u)$ cannot perform any action either since, by the root condition, any possible initial $\tau$-transition from $\sigma_\nil(u)$ would have to be matched by a $\tau$-transition from $\sigma_\nil(t)$.
By Lemma~\ref{lem:emptysubstitutionAct}.\ref{lem:not_act_free}, we can then conclude that $u$ is action free.

Let us now deal with the second item.
First of all, we notice that since $t \approx u$ is provable from $\E$, then it is sound modulo $\sim_\rwb$.
Hence, as $t$ is action-free, we can apply Proposition~\ref{prop:act_free_rwb}.\ref{prop:both_act_free} and obtain that $u$ is action-free as well.
The proof then proceeds by induction on the length of the proof of $t \approx u$ from $\E$, where the inductive step is carried out by a case analysis on the last rule of equational logic that is used in the proof.
We expand only the case in which the last rule applied is an instance of the substitution rule.
The other cases are standard.

Assume that $t = \sigma(t')$ and $u = \sigma(u')$ for some substitution $\sigma$ and CCS terms $t',u'$ such that $t' \approx u' \in \E$.
Since $t$ and $u$ are both action-free, from Lemma~\ref{lem:emptysubstitutionAct}.\ref{lem:not_act_free} we can infer that $t'$ and $u'$ are action-free as well. 
In fact, if $t'$ was not action-free, we could directly infer that $\sigma(t') \trans[\mu]$ for some $\mu \in \Acttau$, thus giving a contradiction with $t = \sigma(t')$ being action-free.  
Similarly for $u'$.
Hence $t' \approx u' \in \E$ is action-free.
Notice now that since $t$ and $u$ are both action-free, we can infer that the substitution rule used in the last step of the proof of $t \approx u$ is action-free.
Therefore, we can conclude that the proof of $t \approx u$ from $\E$ uses only action-free equations.
\end{proof}

Theorem~\ref{thm:rwb_not_fin} can then be obtained as a direct consequence of the following result:

\begin{theorem}
\label{thm:rwb_Mn}
Assume that $\E$ is a finite axiom system over CCS that is sound modulo $\sim_\rwb$. Then there exists some $n \ge 1$ such that $\E \not\vdash M_n$, where $M_n$ is the $n$-th member of the family $\{M_n\}_{n \ge 1}$ introduced in Section~\ref{sec:classic}.
\end{theorem}

\begin{proof}
Assume, towards a contradiction, that $\E \vdash M_n$ for all $n \ge 1$.
Since $M_n$ is an action-free equation for each $n \ge 1$, by Proposition~\ref{prop:act_free_rwb}.\ref{prop:only_act_free} we have that all the equations that are used in the proof from $\E$ of $M_n$ are action-free as well.
Moreover, for each $n \ge 1$, $M_n$ can be proved by using finitely many action-free equations, as $\E$ is finite by the proviso of the theorem. 
By Proposition~\ref{prop:act_free_rwb_b} we have that all these equations are also sound modulo $\sim_\B$ over CCS$\!\!\phantom{.}_{\Act}$. 
Therefore, we can conclude that the finite axiom set $\E$ allows us to prove $M_n$, for all $n \ge 1$, over CCS$\!\!\phantom{.}_{\Act}$.
This contradicts the negative result obtained by Moller in \cite{Mo90a}, and we can therefore conclude that there is at least one $n \ge 1$ such that $M_n$ is not provable from $\E$.
\end{proof}


\section{Back to the future: related and future work}
\label{sec:conclusion}

Due to the central role played by the parallel composition operator in concurrency, it is not surprising that the literature abounds with studies on the equational characterisation of (various versions) of this operator, to the point that any possible list of related works we shall make will be far from exhaustive.
Hence, we decided to focus on works that are close to the three new achievements we have presented in this paper and to our plans for future work.
We refer the interested reader to the textbooks \cite{BW90,Fo00,BBR10} and the survey papers \cite{AFIL05a,AI07} for a discussion of general (and missing) references.

In Section~\ref{sec:no_binary} we have presented the results from \cite{ACFIL21}, that constitute a first step towards a definitive justification of the canonical standing of the left and communication merge operators by Bergstra and Klop. 
A natural direction for future work is then the following:
\begin{equation}
\tag{FW1}
\parbox{\dimexpr\linewidth-4em}{
\strut
\emph{Can we obtain a full answer to problem (\ref{eq:problem})}?
\strut
}
\end{equation}
We envisage the following ways in which we might generalise our contribution in \cite{ACFIL21}. 
Firstly, we will try to get rid of Assumptions~\ref{Ass:equation} and~\ref{Ass:targets}.
Next, it is natural to relax Assumption~\ref{Ass:deSimone} by considering the GSOS format \cite{BIM95} in place of the de Simone format. 
However, we believe that this generalisation cannot be obtained in a straightforward manner and that it will require the introduction of new techniques.
It would also be very interesting to explore whether some version of problem (\ref{eq:problem}) can be solved using existing results from equational logic and universal algebra \cite{BMS13,BS81}. 

Remaining in the context of finite axiomatisations of bisimilarity over CCS, an interesting research question, that is still in need of an answer, is the following:
\begin{equation}
\tag{FW2}
\parbox{\dimexpr\linewidth-4em}{
\strut
\emph{Is there a finite, complete axiomatisation of bisimilarity over full recursion-free CCS}?
\strut
}
\end{equation}
A finite, complete equational axiomatisation of bisimilarity over the recursion, restriction and relabelling free fragment of CCS, extended with the left and communication merge operators is given in \cite{AFIL09}.
Later on, in \cite{AILT08}, finite, complete equational axiomatizations for fragments of CCS with restriction and relabelling were studied, but considering only pure interleaving parallel composition, thus requiring only the addition of the left merge operator to the syntax of CCS.
However, from our preliminary studies, it seems that the result of \cite{AILT08} no longer holds when we consider the fragment of CCS with restriction and the full parallel composition operator, i.e., by allowing also the synchronisation of parallel components.

In \cite{FL00} and \cite{AFIL05a} finite bases for bisimilarity with respect to PA and BCCSP with parallel composition, extended with the auxiliary operators left merge and communication merge were presented.
Furthermore, in \cite{CFLN08} an overview was given of which behavioural equivalences in the linear time-branching time spectrum are finitely based with respect to BCCSP. 

In the literature, we can find some studies on the axiomatisability of behavioural congruences that are not included in the spectrum of \cite{vG90}.
For instance, in \cite{AFIL05b} it is shown that split-$2$ bisimilarity \cite{vGV87}, a bisimulation like relation on algebras in which the execution of each action is split into an observable beginning and ending, affords a finite equational axiomatisation over CCS enriched with Hennessy's merge.
The papers \cite{CF08,CFvG15} deal with the equational axiomatisation of impossible futures \cite{Vo92,VM01} over BCCSP.
They show that while the impossible futures preorder affords a finite equational axiomatisation, there is no finite, complete axiomatisation of impossible futures equivalence over BCCS even if the set of action is a singleton.
Despite all our efforts, we have not managed to lift the results of \cite{CF08,CFvG15} over CCS yet.
A natural question that arises from this problem is the following:
\begin{equation}
\tag{FW3}
\parbox{\dimexpr\linewidth-4em}{
\strut
\emph{Are there general techniques for lifting negative results across process algebras}?
\strut
}
\end{equation}
Understanding whether it is possible to lift non-finite axiomatisability results among different algebras, and under which constraints this can be done, is an interesting research avenue and we aim to investigate it in future work. 
A methodology for transferring non-finite-axiomatisability results across languages was presented in \cite{AFIM10}, where a reduction-based approach was proposed. 
However, that method has some limitations and thus further studies are needed.

In a similar fashion, the original contribution presented in this paper, i.e., the non-finitely axiomatisability result for rooted weak bisimilarity over CCS, paves the way to the following research question:
\begin{equation}
\tag{FW4}
\parbox{\dimexpr\linewidth-4em}{
\strut
\emph{Are there general techniques for lifting negative results from strong to weak congruences}?
\strut
}
\end{equation}
An answer to this question will allow us to solve many problems that have been already solved for strong semantics, but that are still open for weak semantics.

The lifting techniques possibly found from the research questions (FW3) and (FW4) could also help us to develop a third lifting technique, in order to solve the following problem:
\begin{equation}
\tag{FW5}
\parbox{\dimexpr\linewidth-4em}{
\strut
\emph{Are there general techniques for lifting (negative) results to probabilistic languages}?
\strut
}
\end{equation}
The ever increasing interest in probabilistic systems has in fact inspired a number of studies on the axiomatisation of probabilistic congruences.
The substantial differences in the considered probabilistic models (see \cite{De05} for a survey) prevent us from discussing these results in detail.
Yet, we remark that we can find studies on strong probabilistic semantics \cite{AEI02,BBS95,vGSST90,GV19,LS92,Mi20,SS00,TG20}, weak probabilistic semantics \cite{ABW06,AG09,vGGV19}, as well as on metric semantics \cite{dAGL14}.
Further studies in this direction are encouraged by recent achievements on probabilistic branching semantics \cite{CT20a,CT20b} and behavioural metrics \cite{CLT19,CLT20}.

\section*{Acknowledgments}
This work has been supported by the project `\emph{Open Problems in the Equational Logic of Processes}' (OPEL) of the Icelandic Research Fund (grant No.~196050-051).

\bibliographystyle{splncs04}
\bibliography{LICS21_bib}

\end{document}